\documentclass[lettersize,journal]{IEEEtran}
\usepackage[utf8]{inputenc}
\usepackage[T1]{fontenc}
\usepackage{graphicx,color,mathenv}
\usepackage[cmex10]{amsmath}
\usepackage{cite,amsthm,amssymb}
\usepackage{balance}
\usepackage[ruled,vlined]{algorithm2e}
\usepackage{subfigure,amsfonts}
\usepackage{float}
\usepackage{url}
\usepackage{xcolor}

\usepackage{bbm}
\usepackage{cancel}

\DeclareMathAlphabet\mathbfcal{OMS}{cmsy}{b}{n}

\usepackage[all=normal,paragraphs=normal,floats=tight,mathspacing=tight,wordspacing=tight,charwidths=tight,mathdisplays=tight,leading=normal]{savetrees}

\newtheorem{theorem}{Theorem}
\newtheorem{proposition}{Proposition}

\begin{document}

\title{\LARGE On the Robustness of AFDM and OTFS Against Passive Eavesdroppers}

\author{Vincent Savaux, \IEEEmembership{Senior Member,~IEEE},
Hyeon Seok Rou, \IEEEmembership{Member,~IEEE},
Zeping Sui, \IEEEmembership{Member,~IEEE}, \\
Giuseppe Thadeu Freitas de Abreu, \IEEEmembership{Senior Member,~IEEE}, 
Zilong Liu, \IEEEmembership{Senior Member,~IEEE}

\thanks{V.~Savaux is with the Institute of Research and Technology b\textless\textgreater com, 35510 Cesson S{\'e}vign{\'e}, France (email: vincent.savaux@b-com.com).}
\thanks{H.~S.~Rou and G.~T.~F.~de~Abreu are with the School of Computer Science and Engineering, Constructor University Bremen, Campus Ring 1, 28759 Bremen, Germany (email: [hrou, gabreu]@constructor.university).}
\thanks{Z. Sui and Z. Liu are with the School of Computer Science and Electronics Engineering, University of Essex, Colchester CO4 3SQ, U.K. (email: zepingsui@outlook.com, zilong.liu@essex.ac.uk).}
}

\markboth{To Be Submitted To IEEE Wireless Communications Letters}%
{Shell \MakeLowercase{\textit{et al.}}: A Sample Article Using IEEEtran.cls for IEEE Journals}


\maketitle

\begin{abstract}
We investigate the robustness of affine frequency division multiplexing (AFDM) and orthogonal time frequency space (OTFS) waveforms against passive eavesdroppers performing brute-force demodulation to intercepted signals, under the assumption that eavesdroppers have no knowledge of chirp parameters (in AFDM) or the delay–Doppler grid configuration (in OTFS), such that they must search exhaustively over possible demodulation matrices. 
Analytical results show that the brute-force complexity scales as $\mathcal{O}(\sqrt{N})$ for OTFS and $\mathcal{O}(N^2)$ for AFDM, where $N$ is the number of subcarriers, indicating that AFDM has superior resilience over OTFS.
Bit error rate (BER) simulations confirm the analysis by showing that, with AFDM, the signal remains nearly undecodable at the eavesdropper, while OTFS allows partial signal recovery under equivalent conditions.
\end{abstract}



\section{Introduction}
\label{sec:intro}

Network security has often been restricted to complex key-based encryption schemes. However, keyless techniques, generally based on the signal-to-interference-plus-noise ratio (SINR) optimization, have recently emerged as promising approaches \cite{hamamreh19,wang19}, complementing traditional key-based methods. 
These methods usually rely on the use of spatial diversity, and relatively few studies address physical layer security (PLS) PLS inherent to the waveform itself.

As a parallel development, there has been a growing interest in affine frequency division multiplexing (AFDM) \cite{bemani2023affine,rou2025stcommag} and orthogonal time frequency space (OTFS) \cite{wei2021orthogonal, hadani17} over the past few years as two alternatives to orthogonal frequency division multiplexing (OFDM) for next-generation communication systems, as they better cope with doubly dispersive channels \cite{rou24,sui2025multi}.
The performance of both modulation schemes has been extensively compared in terms of bit error rate (BER), peak-to-average power ratio (PAPR), and capability of supporting integrated integrated sensing and communication (ISAC) \cite{yuan2022orthogonal,ranasinghe2024joint}, in addition to among other functionalities such as index modulation \cite{sui2025generalized,10129061,rou2024cpimafdm}.
However, their robustness against threats in the context of physical layer (PHY) security remains a largely open topic. 

For example, in AFDM, by leveraging permutations over the chirp sequences \cite{Rou_Arxiv25_CPAFDM}, a novel PHY security approach was presented in \cite{rou25afdmpls} which was shown to be virtually perfectly secure even against quantum-accelerated eavesdroppers due to the immense complexity in the combinatorial space.
Alternatively, techniques to realize PHY security over the conventional AFDM parameters (usually denoted as $c_1$ and $c_2$) are reported in \cite{Wang_ICC25,Tek_TVT25,Chen_ICC25,di25}, via parameter hopping of $(c_1,c_2)$. 
In \cite{Wang_ICC25}, a pseudo random sequence is selected for the pre-chirp parameter $c_2$ from a codebook, whereas in \cite{Tek_TVT25}, $c_2$ is securely generated at both legitimate transmitter and receiver sides based on their common hidden communication channel. 
In both studies, the security enhancement is based on the parameter $c_2$ only. Nevertheless, $c_1$ can also play a crucial role in PLS. In \cite{Chen_ICC25,di25}, the authors analyzed the range of $(c_1,c_2)$ to guarantee security and then evaluated the robustness of AFDM through simulations only. 
Similarly, based on the inherent parameters of the OTFS waveform, delay-Doppler precoding was proposed in \cite{liu24} for security enhancement of OTFS. However, to the best of our knowledge, beyond methods that improve the PHY security of AFDM and OTFS \cite{Wang_ICC25,Tek_TVT25,Chen_ICC25,di25,liu24}, no study has been evaluated the intrinsic robustness of these waveforms against eavesdropping to date.   

Therefore, in this paper, we analyze and compare the robustness of AFDM and OTFS modulations against eavesdroppers attempting to brute-force demodulate the leaked signals they receive from legitimate users. 
We consider a malicious user who performs blind demodulation by exhaustively testing all the possible modulation parameters, \emph{i.e.} the delay-Doppler grid size $(K,L)$ in OTFS, and the chirp parameters $(c_1,c_2)$ in AFDM, given that the number $N$ of subcarriers is known. 
We then assess the robustness in terms of the maximum number of attempts an eavesdropper has to perform to demodulate the signal properly. 
We show that the brute-force complexity of OTFS and AFDM scales as $\mathcal{O}(\sqrt{N})$ and $\mathcal{O}(N^2)$, respectively.
This analysis proves that AFDM is significantly more robust than OTFS against brute-force demodulation, because the chirp parameters $(c_1,c_2)$ are chosen within a continuous subset of $\mathbb{R}^2$. In contrast, the delay-Doppler grid size $(K,L)$ in OTFS corresponds to the limited number of divisors of $N$. 
It should be noted that the proposed study may serve as a reference to all parameter hopping-based methods \cite{Wang_ICC25,Tek_TVT25,Chen_ICC25} to assess their robustness, besides numerical results theoretically. 
Furthermore, simulation results validate the theoretical developments on the robustness and show the superiority of AFDM over OTFS, where it is shown that the BER of AFDM at the eavesdropper remains flat and undecodable for any SNR range, while it converges to a moderate signal recovery for OTFS, given the same number of demodulation attempts.

%
%
%
%
%
%


The rest of the paper is organized as follows: Section \ref{sec:system_model} presents both the AFDM and OTFS signal models and the eavesdropping scenario. The theoretical robustness analysis is developed in Section \ref{sec:robustness}. Simulation results are provided in Section \ref{sec:simu}, and conclusions are drawn in Section \ref{sec:conclusion}.   



\section{System Model}
\label{sec:system_model} 

This section introduces the signal models of both AFDM and OTFS modulations, as well as the scenario considered in this paper, involving legitimate users and a co-located eavesdropper attempting to demodulate the communications of the former. 


\subsection{AFDM and OTFS Signals Models}

Let us consider a multicarrier AFDM or OTFS signal consisting of $N$ subcarriers. Inspired by \cite{boudjelal25,savaux20256GNetafdm}, both waveforms can be interpreted as a precoded OFDM modulation scheme, whose transmitted signal can be expressed as  
\begin{equation}
    \textbf{x} = \mathbfcal{F}_{N}^H \textbf{Q} \textbf{d}, 
    \label{eq:xgen}
\end{equation} 
where the vector $\textbf{d} \in \mathbb{C}^N$ contains the data randomly taken from a constellation, and the matrix $\mathbfcal{F}_{N}^H$ is the IDFT matrix of size $N\times N$ containing the element $\frac{1}{\sqrt{N}}e^{2j\pi \frac{mn}{N}}$ at entry $(n,m)$. Then, the expression of the recoding matrix $\textbf{Q}$ depends on the considered modulation. For instance, we simply have $\textbf{Q} = \textbf{I}_N$ in OFDM. In AFDM, it is given by 
\begin{equation}
    \textbf{Q} = \mathbfcal{F}_{N}\boldsymbol{\Lambda}_{c_1} \mathbfcal{F}_{N}^H \boldsymbol{\Lambda}_{c_2}, 
    \label{eq:Qafdm}
\end{equation}
where $\boldsymbol{\Lambda}_{c_i} = \text{diag}([e^{2 j \pi c_i 0^2},..,e^{2 j \pi c_i (N-1)^2}]) \in \mathbb{C}^{N \times N}$, $i=1,2$, with $c_1, c_2 \in \mathbb{R}$ the so-called chirp parameters which can be set to achieve full diversity in time and frequency selective channels \cite{bemani2023affine,rou24}, as described in the following.

In OTFS, the matrix $\textbf{Q}$ is given by 
\begin{equation}
    \textbf{Q} = \mathbfcal{F}_{N}(\mathbfcal{F}_{L}^H \otimes \textbf{I}_K), 
    \label{eq:Qotfs}
\end{equation}
where $K \times L$ is the size of the delay-Doppler grid the data is mapped on. Interestingly, note that in OTFS, the delay-Doppler diversity is directly dependent on $K$ and $L$, and in turn on $N$ since $N=K \times L$. In contrast, it is related to $c_1$ and $c_2$ in AFDM, independently of $N$. 

Omitting the cyclic prefix (CP) or chirp-periodic prefix (CPP) addition and removal, the general input-output relation in SISO systems can be expressed as 
\begin{equation}
    \textbf{y} = \underbrace{\sum_{l=0}^{L_c-1}h_l \boldsymbol{\Delta}_{\theta_l}\boldsymbol{\Pi}^l \textbf{x}}_{\textbf{H} \textbf{x}}  + \textbf{w}, 
    \label{eq:receivedsignaly}
\end{equation}
where $\textbf{y} \in \mathbb{C}^{N \times 1}$ is the vector of the received signal, and $\textbf{w} \in \mathbb{C}^{N \times 1}$ is the vector of the additive white Gaussian noise with independent and identically distributed samples $w_n \sim \mathbb{C}\mathcal{N}(0,\sigma^2)$.

In turn, $\textbf{H} \in \mathbb{C}^{N \times N}$ is the channel matrix, where  $h_l$ is the $l$th channel path coefficient (possibly null if the channel is sparse), and $L_c$ is the channel length. Moreover, $\boldsymbol{\Delta}_{\theta_l} \in \mathbb{C}^{N \times N}$ is the diagonal matrix containing the samples $e^{2j\pi \frac{\theta_l n}{N}}$, where $\theta_l \in [0,\theta_{\max}]$ is the normalized Doppler shift (integer of fractional component), and $\theta_{\max}$ the maximum normalized Doppler shift. Then, $\boldsymbol{\Pi}$ is the forward cyclic-shift matrix. Note that the full-diversity property of AFDM holds if \cite{bemani21,rou24}: 
\begin{equation}
      \frac{\theta_{\max}}{N} \leq c_1 , \hspace{1cm} c_2 << \frac{1}{N}. 
    \label{eq:conditionsafdm}
\end{equation}
Furthermore, we note without loss of generality that $c_1 << 1$, which physically means that the number of chirps per symbol is significantly lower than the number of samples per symbol. Note that more details on how $c_1$ should be chosen in the context of PLS can be found in \cite{Chen_ICC25,di25}. 


\subsection{Scenario}

\begin{figure}[tbp]
\centering{\includegraphics[width=0.8\columnwidth]{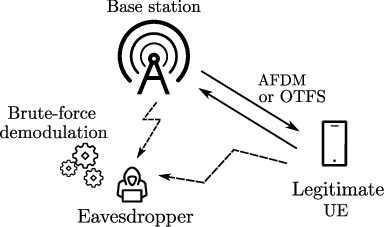}}
  \caption{Considered scenario in which passive eavesdroppers try to brute-force demodulate leaked communications links between a base station and legitimate UEs.}
\label{fig:scenario}
\vspace{-0.5em}
\end{figure}

We consider a scenario in which base stations (BSs) communicate with legitimate UEs of the network, while malicious eavesdroppers attempt to brute-force demodulate the leaked signals, as illustrated in Fig. \ref{fig:scenario}. These signals are assumed to be modulated using AFDM or OTFS, as previously described, and the eavesdropper is aware of the waveform it receives. Note that, even though we consider a cellular system comprising BS and UEs, the model is general enough to be extended to any type of communication, such as cell-free, device-to-device (D2D), vehicular-to-everything (V2X), side link, or mutli-input multi-output (MIMO) systems.

We deliberately assume a worst-case scenario (from the point of view of the legitimate stakeholders of the network) in which the eavesdropper is synchronized with the leaked signals it receives, and has a perfect knowledge of the channel $\textbf{H}$ between the transmitter (\emph{e.g.} a BS or a UE) and itself, as well as the number of subcarriers $N$. 
Note that this therefore limits CSI-based PHY security approaches \cite{Tek_TVT25} which depends on the hidden CSI from the eavesdroppers.
It results in the capability of perfect equalization, such that the equalized signal at the eavesdropper can be expressed as: 
\begin{equation}
    \hat{\textbf{x}} = \textbf{G}\textbf{y} = \textbf{x} + \textbf{G}\textbf{w}, 
    \label{eq:equalization}
\end{equation}
where $\textbf{G} \in \mathbb{C}^{N \times N}$ is the equalization matrix such as $\textbf{G}\textbf{H}=\textbf{I}_N$. The demodulation then performs the data recovery expressed as $\hat{\textbf{d}} = \textbf{Q}^{-1}\mathbfcal{F}_{N} \hat{\textbf{x}}$. In contrast, it is assumed that the modulation parameters (\emph{i.e.} $c_1$ and $c_2$ in AFDM, and $K$ and $L$ in OTFS), and in turn the decoding matrix $\textbf{Q}^{-1}$, are unknown to the eavesdropper. Consequently, the latter adopts a brute-force demodulation strategy to estimate the transmitted data $\textbf{d}$, which means that it exhaustively tests all the possible values of the modulation parameters until it recovers the data. Despite the exhaustive search may seem to be an oversimplified method, it is optimal in the sense of the maximum likelihood in blind estimation of unknown parameters. In the following, we analyze the robustness of OTFS and AFDM against such a brute-force demodulation.  


\section{Robustness Analysis}
\label{sec:robustness} 

In this section, we evaluate the robustness of the OTFS and AFDM modulation schemes in terms of the maximum number of attempts, denoted by $M_a$, that an eavesdropper should carry out to recover the transmitted data $\hat{\textbf{d}}$ from $\hat{\textbf{x}}$ in (\ref{eq:equalization}). It is worth emphasizing that we assess the inherent robustness of the waveforms, independently of additional secure methods as presented in \cite{Wang_ICC25,Tek_TVT25,Chen_ICC25,di25,liu24}, or independently of any other techniques that aim to secure the transmission at the data level, \emph{e.g.} by applying a pseudo-random matrix to $\textbf{d}$ directly. Since we focus on the robustness of the waveform against brute-force demodulation, we can deliberately omit the noise in this section, \emph{i.e.} $\textbf{w} \rightarrow \textbf{0}$ in (\ref{eq:equalization}). In other words, $M_a$ corresponds to the maximum attempts the eavesdropper should make to find the decoding matrix $\textbf{Q}'^{-1}$ such as $\textbf{Q}'^{-1}\textbf{Q}=\textbf{I}_N$, with a probability of $1$, given that $N$ is known. The larger the number $M_a$, the stronger the waveform. 


\subsection{OTFS}

In OTFS, the brute-force strategy consists in testing all the possible decoding matrices $\textbf{Q}^{-1}$ according to $K$ (or $L$ equivalently), and keep the set of matrices $\textbf{Q}'^{-1}$ (parametrized by $(K',L')$) leading to $\textbf{Q}'^{-1}\textbf{Q}=\textbf{I}_N$. The robustness of OTFS is given in Proposition \ref{prop:otfs}. 

\begin{proposition} 
\label{prop:otfs}
Given an OTFS signal composed of $N$ subcarriers and parametrized by $(K,L)$, its robustness against brute-force demodulation is given by 
\begin{equation}
    M_a = \sigma(N) \leq 2\sqrt{N},  
    \label{eq:Maotfs}
\end{equation}
where $\sigma(N)$ is the number of integer divisors of $N$. 
\end{proposition}

\begin{proof}
Note that this result is a direct consequence of the so-called Dirichlet divisor problem, whose more precise upper bounds can be found in \cite{huxley03}. First, we can readily show from (\ref{eq:Qotfs}) that the solution to $\textbf{Q}'^{-1}\textbf{Q}=\textbf{I}_N$ is unique and is given by $(K',L')=(K,L)$. Thus, the number of attempts performed by an eavesdropper directly depends on the number of integer divisors of $N = K \times L$, due to the rectangular delay-Doppler grid structure of OTFS.
Then, since the number of divisors of $N$ is twice the number of divisors of $N$ between 1 and $\sqrt{N}$, $M_a$ can be upper-bounded by $2\sqrt{N}$, which concludes the proof. 
\end{proof}

We deduce from (\ref{eq:Maotfs}) that the robustness of OTFS exhibits an inverse quadratic growth with respect to the number of subcarriers, which then becomes weak for low $N$ values. Furthermore, in practice, $M_a$ in (\ref{eq:Maotfs}) largely overestimates the possible number of solutions for $\textbf{Q}^{-1}$ because: 
i) we know from Dirichlet that the average number of divisors $\sigma(N)$ of $N$ is rather asymptotically equal to $\ln(N) +2\gamma -1 \leq 2\sqrt{N}$, where $\gamma$ is the Euler-Mascheroni constant i.e., $M_a = 2\sqrt{N}$ is a loose upper bound, and ii) we know from \cite{rou25afdmpls} that only a subset of all possible values $(K,L)$ should be considered to guarantee delay-Doppler diversity (\emph{e.g.} we know that if $K=1$, OTFS is exactly equivalent to OFDM).


\subsection{AFDM}

In AFDM, the brute-force demodulation \emph{a priori} involves an exhaustive joint search of $(c_1,c_2)$ in a continuous subset $\Omega_c \in \mathbb{R}^2$, which is theoretically impractical within a reasonable time. In fact, $M_a$ tends to infinity due to the continuous nature of $\Omega_c$. However, to derive a finite value of $M_a$, we can first notice that the demodulation through $\textbf{Q}^{-1}\mathbfcal{F}_{N} =   \boldsymbol{\Lambda}_{c_2}^*\mathbfcal{F}_{N}\boldsymbol{\Lambda}_{c_1}^*$ is a sequence of three distinct operations: the dechirp in the time domain, then the DFT, and finally the dechirp in the frequency domain. It follows that $c_1$ and then $c_2$ can be sequentially tested. Moreover, under the assumption that $c_1$ is known, an error on $c_2$ only induces a phase rotation in the demodulated data $\hat{\textbf{d}}$, which can be easily estimated and corrected, for instance by the maximum lakelihood estimator. 

In contrast, an error on $c_1$ prevents the possible demodulation in general. Thus, $c_1$ is the most limiting parameter in the exhaustive search of $(c_1,c_2)$ allowing for a brute-force demodulation by an eavesdropper. For this reason, in this section, we focus on $c_1$, although $c_2$ can also play a role in PHY security, as shown in \cite{Wang_ICC25,Tek_TVT25,Chen_ICC25,di25}. 

To evaluate $M_a$, we characterize how far from $c_1$ a given value $c_1'$ tested by the eavesdropper must be to fail to demodulate the signal. To this end, we can restrict the possible range of $c_1'$ to $c_1' \in [\frac{\theta_{\max}}{N},D]$, where $D<<1$, according to the conditions given in (\ref{eq:conditionsafdm}). Then, we express the robustness of AFDM in Theorem \ref{theor:afdm}. 

\begin{theorem}
\label{theor:afdm}
Given an AFDM signal composed of $N$ subcarriers and parametrized by $(c_1,c_2)$, its robustness against brute-force demodulation is given by 
\begin{equation}
    M_a = \frac{1}{N}\left[\pi(DN-\theta_{\max})(N-1)^2\right]. 
    \label{eq:MaAFDM}
\end{equation}
\end{theorem}
\begin{proof}
First, we express the samples $x_n$ of $\textbf{x}$ in (\ref{eq:xgen}) for any $n=0,1,N-1$ as 
\begin{equation}
    x_n = \frac{1}{\sqrt{N}} \sum_{m=0}^{N-1}d_m e^{2j\pi(c_1 n^2 + c_2 m^2 + \frac{mn}{N})}, 
    \label{eq:xn}
\end{equation}
where $m$ is the subcarrier index, and $d_m$ is the $m$th element of the vector $\textbf{d}$. 

By assuming that the eavesdropper attempts to brute force the received AFDM signal using $(c_1',c_2')$, the sample $\hat{d}_k$ of $\hat{\textbf{d}}$, $k=0,1,..,N-1$, is given by the DAFT of $x_n$ as 
\begin{equation}
    \hat{d}_k = \frac{1}{\sqrt{N}} \sum_{n=0}^{N-1} x_n e^{-2j\pi(c_1' n^2 + c_2' k^2 + \frac{kn}{N})}.  
    \label{eq:hdm1}
\end{equation}

Then, by substituting (\ref{eq:xn}) into (\ref{eq:hdm1}), and defining $\Delta_1 = c_1-c_1'$ and $\Delta_2 = c_2 - c_2'$, we obtain: 
\begin{align}
    \hat{d}_k =& \frac{1}{N} \sum_{n=0}^{N-1} \sum_{m=0}^{N-1} d_m e^{2j\pi((c_1-c_1') n^2 + c_2m^2-c_2' k^2 + \frac{(m-k)n}{N})} \nonumber \\
    =& \frac{e^{2j\pi\Delta_2k^2}}{N} \sum_{m=0}^{N-1} d_m e^{2j\pi c_2(m^2-k^2)} \underbrace{\sum_{n=0}^{N-1} e^{2j\pi(\Delta_1 n^2 + \frac{(m-k)n}{N})}}_{S_1}, 
    \label{eq:hdm2}
\end{align}
where $S_1$ is defined for clarity.

It must be noted that $S_1$ simplifies only for integer $\Delta_1$ values, since it reduces to the sum of the $N$th roots of unity.
In that case $S_1 = N \delta_{m,k}$ where $\delta_{m,k}$ is the Kronecker delta, and hence $\hat{d}_k = e^{2j\pi\Delta_2k^2} d_k$, where the term $e^{2j\pi\Delta_2k^2}$ corresponds to the aforementioned phase rotation. However, the equality $S_1 = N \delta_{m,k}$ corresponds to either the perfect match $c_1' = c_1$, which is unattainable because $c_1$ is a continuous variable, or $\Delta_1 \in \mathbb{N}^*$ and then $c_1' \notin [\frac{\theta_{\max}}{N},D]$, which is inconsistent with the assumption $c_1' \in [\frac{\theta_{\max}}{N},D]$.     

Despite the equality $c_1' = c_1$ (\emph{i.e.} $\Delta_1 = 0$) cannot theoretically be achieved, the eavesdropper can demodulate the received signal with an acceptable error $|\Delta_1| <<1$, which can be characterized by: 

\begin{equation}
    \left| e^{2j\pi\Delta_1 n^2} -1 \right| \leq |\varepsilon|, 
    \label{eq:error}
\end{equation}
where $|\varepsilon|<< 1$ highlights the acceptable error level, depending on other modulation parameters, but not dealt with in this paper (\emph{e.g.} we can reasonably assume that $|\varepsilon|$ can take a higher value for a low modulation and coding schemes (MCS) than high MCS).  Since (\ref{eq:error}) should hold for any $n=0,1,..,N-1$, and $\Delta_1<<1$, we use the series expansion of the exponential function to derive the upper bound of $\Delta_1$ with $n=N-1$ as

\begin{equation}
    (\ref{eq:error}) \Leftrightarrow |\Delta_1| \leq \frac{|\varepsilon|}{2 \pi (N-1)^2} \leq \frac{1}{2 \pi (N-1)^2},  
    \label{eq:error2}
\end{equation}
where ``1'' can replace $|\varepsilon|$ to obtain an upper bound that only depends on $N$.

Notice that in this case, the demodulated data becomes $\hat{d}_k = e^{2j\pi\Delta_2k^2} d_k (1 + \varepsilon)$. In AFDM, the brute-force strategy then consists in testing all possible values $c_1'$ within $[\frac{\theta_{\max}}{N},D]$ with a step of $2|\Delta_1| = \frac{1}{\pi (N-1)^2}$. Then, the maximum number of attempts $M_a$ corresponds to the ratio of the range of the search set $[\frac{\theta_{\max}}{N},D]$ and the step $2|\Delta_1|$, leading to (\ref{eq:MaAFDM}), which concludes the proof.  
\end{proof}

We deduce from (\ref{eq:MaAFDM}) that the robustness of AFDM is quadratically proportional to the number of subcarriers $N$ and also linearly proportional to the size of the search set through $D$. Thus, AFDM should be much stronger than OTFS against passive eavesdroppers, therefore limiting a brute-force demodulation of a leaked signal in real-time. Unlike OTFS, $M_a$ in (\ref{eq:MaAFDM}) is independent to chirp parameters $c_1$ and $c_2$, hence AFDM achieves both robustness against eavesdroppers and full delay-Doppler diversity. Furthermore, note that $c_2$ must also be dechirped on top of $c_1$ in AFDM to complete the demodulation, thus strengthening the waveform, as addressed in \cite{Wang_ICC25,Tek_TVT25,Chen_ICC25}.

\section{Simulation Results}
\label{sec:simu}

In this section, we validate the theoretically derived results through simulations and we evaluate the performance of the AFDM and OTFS waveforms in terms of the achievable BER at the eavesdropper. In all simulations, an uncoded quadrature phase shift keying (QPSK) modulation has been considered. Simulations have been performed using MATLAB, and the results have been averaged over at least $10^3$ independent Monte-Carlo runs. 

\begin{figure}[tbp]
\centering
  \subfigure[BER versus $K'$ in OTFS, for $N=64,128$ and $K=16,64$. It can be observed that the BER is minimum for $K'=K$. ]
  {\includegraphics[width=\columnwidth]{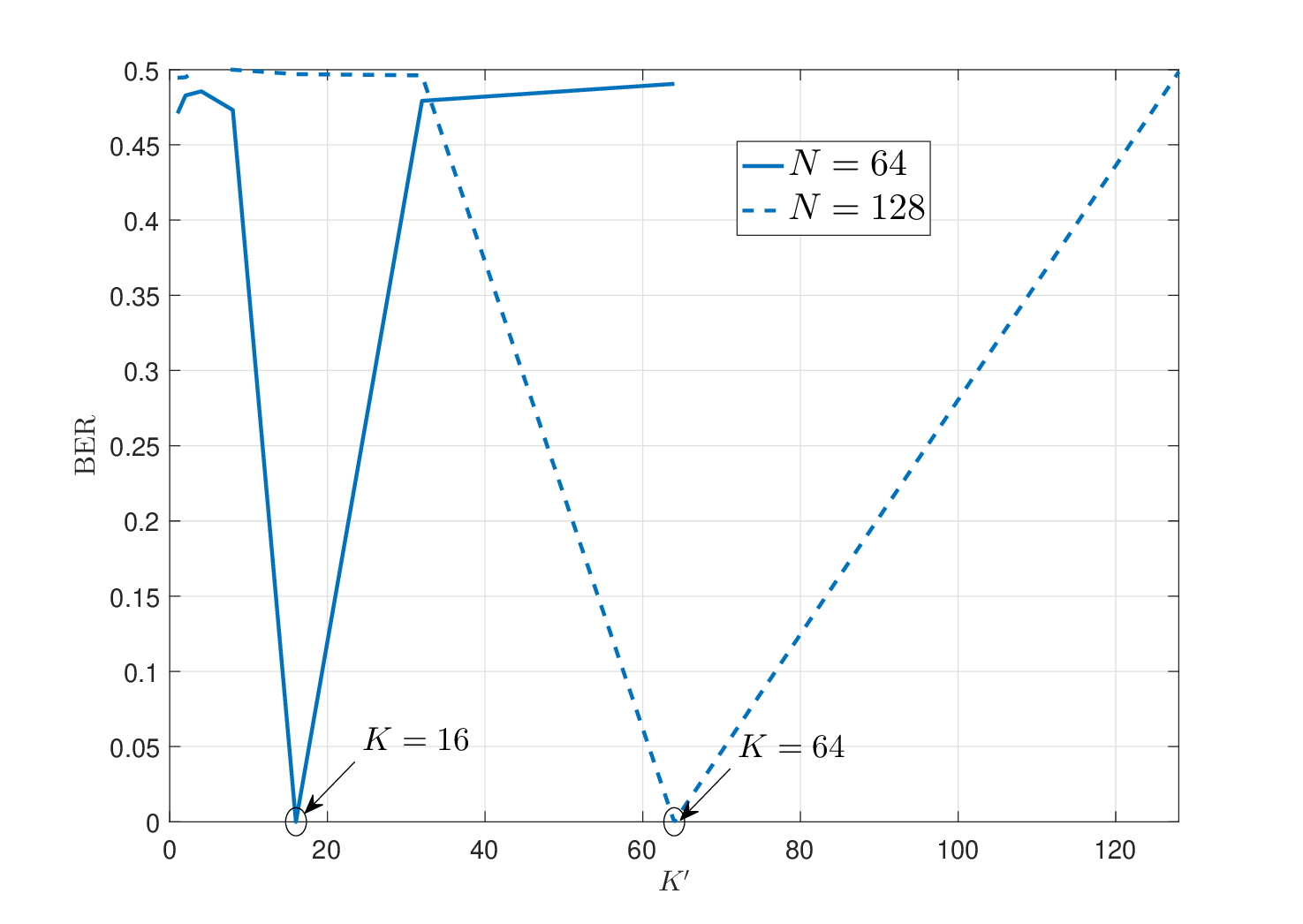}}\quad
  \subfigure[BER versus $c_1'$ in AFDM, for $N=64,128$ and $c_1 = 0.2$. It can be observed that the value of $\Delta_1$ matches the theoretical one.]
  {\includegraphics[width=\columnwidth]{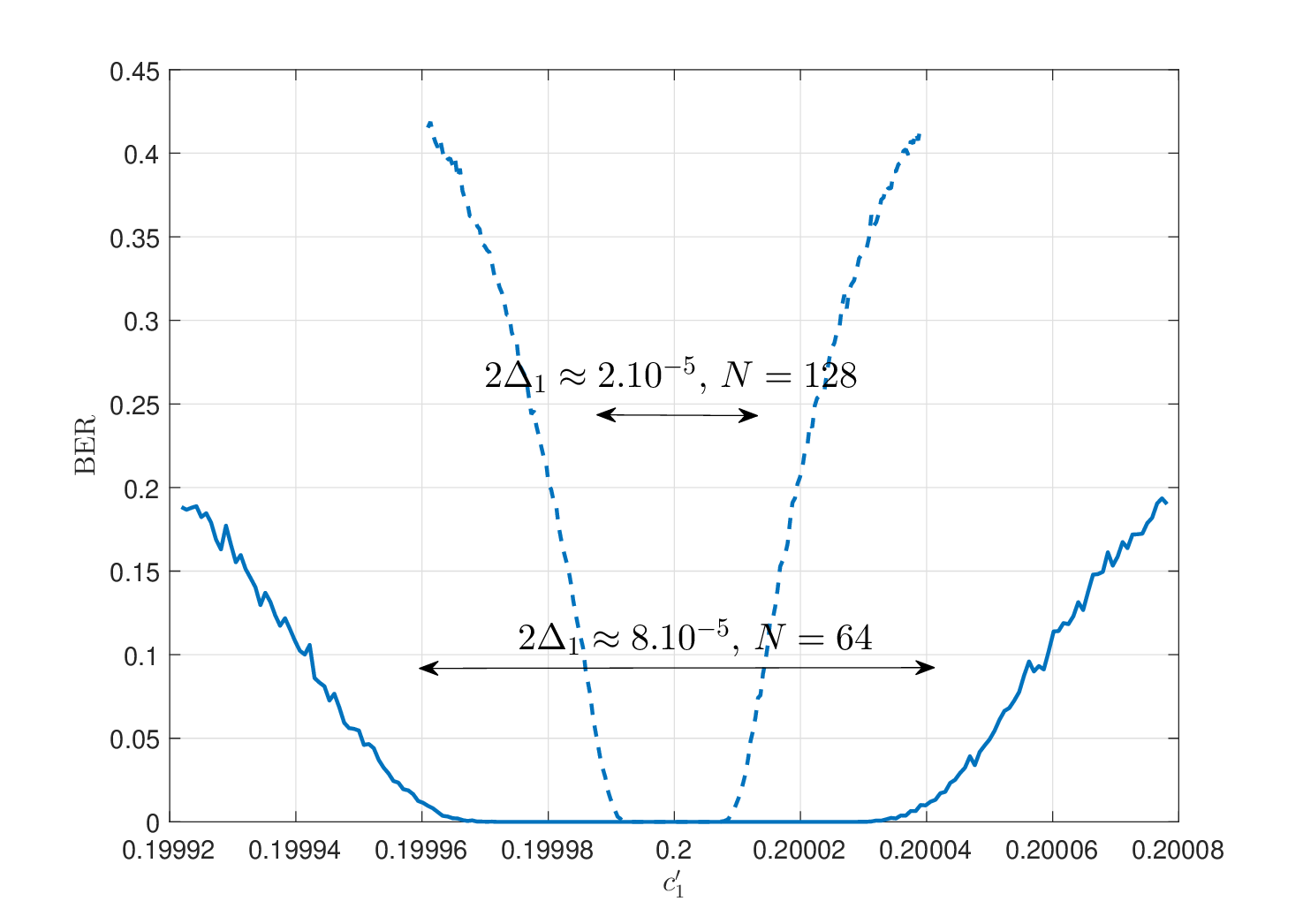}}
\caption{BER versus demodulation parameters $K'$ in (a) OTFS and (b) $c_1'$ in AFDM, for $N=64,128$; SNR=$25$ dB, and QPSK modulation.}
\label{fig:bervsparam}
\vspace{-1em}
\end{figure}

 Fig. \ref{fig:bervsparam} shows the BER versus the demodulation parameters $K'$ in OTFS (a), and $c_1'$ in AFDM (b), for $N\in \{64,128\}$ and $K \in \{16,64\}$, and in an additive white Gaussian noise (AWGN) environment such as SNR$=25$~dB. In Fig. \ref{fig:bervsparam}(a), the OTFS parameter $K$ is set to $K=16$. The cardinality of the set to be tested by the eavesdropper to properly demodulate the OTFS signal is $\sigma(\{64,128\})=\{7,8\}$. We observe that the only solution that minimizes the BER is $K'=K$, which confirms the weakness of OTFS against brute-force demodulation. In Fig. \ref{fig:bervsparam}(b), the AFDM parameters $c_1$ and $c_2$ are set to $c_1=0.2$ and $c_2=10^{-3}$, respectively. We assume that the eavesdropper \emph{a priori} knows $c_2$ or can properly estimate this pre-chirp parameter, so it is omitted in the simulations. Thus, it focuses on brute-force demodulation by testing $c_1'$ values. The range of $c_1'$ values has been arbitrarily restricted around $c_1 \pm 8.10^{-5}$ in Fig. \ref{fig:bervsparam}(b) for the sake of clarity. We can observe that the BER reaches a minimum value that spans over $2\Delta_1$ defined in (\ref{eq:error2}). Moreover, other series of simulations show that: i) no other local minimum is achieved when $c_1'$ varies in a wider range of values, and ii) the BER variations become sharper as the constellation size increases. This validates the analysis leading to the upper-bound in (\ref{eq:error2}) and in turn (\ref{eq:MaAFDM}).

To further validate the robustness analysis, Fig. \ref{fig:bervssnr} shows the BER versus SNR (dB) achieved at the eavesdropper for both AFDM and OTFS modulations. Two different setups are considered: AWGN and multipath channel, both using $N=128$ subcarriers. In the latter case, a four-tap channel with $\theta_{\max}=0.3$ is considered, which is equalized at the receiver side using the MMSE equalizer. To fairly compare OTFS and AFDM against brute-force demodulation, we assume that the eavesdropper can test $M_a = \sigma(128) = 8$ values of $K'$ in OTFS and the same number of different $c_1'$ values in AFDM. In the latter case, $c_1'$ is randomly chosen within the set $[\frac{\theta_{\max}}{N},0.3]$, with a minimum spacing of $2\Delta_1$ between two attempts. It can be observed that the BER of OTFS converges to zero when the SNR increases in both channel models, reflecting that $M_a = \sigma(128) = 8$ corresponds to the exhaustive possible values of $K'$, so that the eavesdropper inevitably properly demodulates the OTFS signal. In contrast, the BER of AFDM keeps a value of about $0.5$, which shows that $8$ attempts are largely insufficient to brute-force demodulate the AFDM signal. In fact, more than $10^4$ attempts would be necessary. This ultimately proves that AFDM is much stronger than OTFS against passive eavesdropping.     

\begin{figure}[tbp]
\centering{\includegraphics[width=\columnwidth]{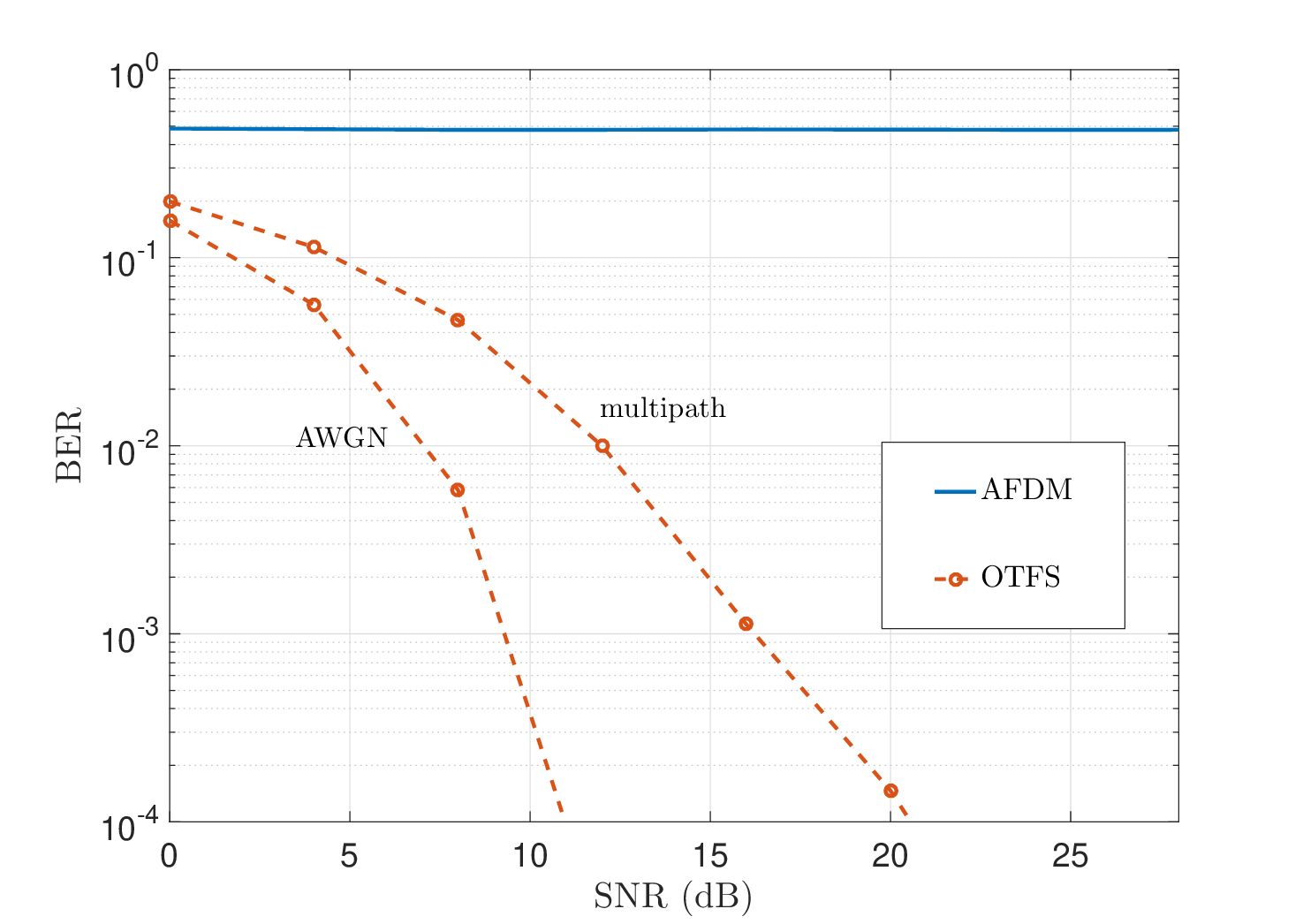}}
  \caption{BER versus SNR (dB) achieved at eavesdropper considering OTFS and AFDM, with $N=128$ in AWGN and multipath channels, using attempts $M_a = \sigma(128) = 8$.}
\label{fig:bervssnr}
\vspace{-1em}
\end{figure}




\section{Conclusion}
\label{sec:conclusion} 

We investigated the robustness of both AFDM and OTFS modulations against eavesdropping, in terms of the maximum number of attempts required for a passive eavesdropper to demodulate the signals via brute-force search, in the absence of any additional PLS method.
It was shown that, for a signal composed of $N$ subcarriers, the corresponding complexity scales as $\mathcal{O}(\sqrt{N})$ and $\mathcal{O}(N^2)$ for OTFS and AFDM, respectively.
This result is due to the nature of the modulation parameters: in OTFS, the delay-Doppler size $(K,L)$ is chosen within the divisors of $N$, whereas in AFDM, the chirp parameters $(c_1,c_2)$ are selected within a continuous subset of $\mathbb{R}^2$.
The analysis, which can also be used to assess the robustness performance of PLS techniques in AFDM and OTFS, indicates that AFDM has an inherent advantage over OTFS in terms privacy.
Simulation results were shown which validated the theoretical analysis and confirmed through the BER that AFDM significantly outperforms OTFS in terms of PHY security.  



\bibliographystyle{IEEEtran}
\bibliography{mabiblio}
\end{document}